\documentclass[12pt]{article}

\usepackage{amsmath,amssymb,amsthm}
\usepackage{graphics,epsfig,calc}
\usepackage{amsmath}
\usepackage{latexsym,epsfig,bm,amssymb}
\usepackage{color}
\usepackage{amsthm,mathrsfs}
\usepackage{mathptmx}

 \newcommand{\beqn}{\begin{eqnarray}}     
 \newcommand{\eeqn}{\end{eqnarray}}     
 \newcommand{\be}{\begin{equation}}     
 \newcommand{\ee}{\end{equation}}     
 \newcommand{\ba}{\begin{array}}     
 \newcommand{\ea}{\end{array}}     
 \newcommand{\pa}{\partial}     
 \newcommand{\re}{\ref}     
 \newcommand{\ci}{\cite}     
 \newcommand{\la}{\label}     
 \newcommand{\bfr}{\begin{flushright}}     
 \newcommand{\efr}{\end{flushright}}     
 \newcommand{\bfl}{\begin{flushleft}}     
 \newcommand{\efl}{\end{flushleft}}     
 \newcommand{\fr}{\frac}     
 \newcommand{\ov}{\overline}     
\newcommand{\Om}{\Omega}

\newcommand{\na}{\nabla}
  
\def\o{\mathaccent"7017}
\newcommand{\Ho}{\o{H}}

\newcommand{\cA}{\cal A}
\newcommand{\cE}{\cal E}

\newcommand{\cS}{\cal S}

\newcommand{\om}{\omega}
\newcommand{\al}{\alpha}   
\newcommand{\si}{\sigma}     
\newcommand{\lr}{\longrightarrow}

\newcommand{\st}{\stackrel}
\newcommand{\tocEF}{\st{{\cal E}_F}\lr}

\newcommand{\tocF}{\st{{\cal F}}\lr}

 \newcommand{\ve}{\varepsilon}    
 \newcommand{\vp}{\varphi}  
 \newcommand{\de}{\delta} 
\newcommand{\ds}{\displaystyle} 
 \newcommand{\De}{\Delta}     
      
 \newcommand{\br}{|\kern-.25em|\kern-.25em|}     
    
      \def\N{{\rm I\kern-.1567em N}}                              
 \def\R{{\rm I\kern-.1567em R}}                              
 \def\C{{\rm C\kern-4.7pt                                    
 \vrule height 7.7pt width 0.4pt depth -0.5pt \phantom {.}}}     
 \def\Z{{\sf Z\kern-4.5pt Z}}                                

 \newtheorem{theorem}{Theorem}[section]      
      
\newtheorem{definition}[theorem]{Definition}     
     
 \newtheorem{lemma}[theorem]{Lemma}

 \newtheorem{remarks}[theorem]{Remarks}     
 \newtheorem{cor}[theorem]{Corollary}     
 \newtheorem{pro}[theorem]{Proposition}

      \newcommand{\brs}{\begin{remarks}}
 \newcommand{\ers}{\end{remarks}}

\begin{document}   

 \begin{titlepage}     
 \begin{center}     

\hfill {\it To the memory of Vladimir Buslaev} 
\vspace{3cm}

 {\Large\bf On global attractors and radiation damping for 
  \medskip\\     
 nonrelativistic particle coupled to scalar field}\\     
  \vspace{1cm}
{\large A. Komech}
\footnote{
The research was carried out at the IITP RAS at 
the expense of the Russian Foundation for Sciences 
(project 14-50-00150).
}
\\
{\it Faculty of  Mathematics Vienna University\\
and Institute for Information Transmission Problems RAS}\\
 Email:~alexander.komech@univie.ac.at
    \\ 
~\\
  {\large E. Kopylova}$\,^1$
\\
{\it Faculty of  Mathematics Vienna University\\
and Institute for Information Transmission Problems RAS}\\
Email:~elena.kopylova@univie.ac.at
\\
~\\
{\large H. Spohn}
\\
{\it Faculty of Mathematics, TU Munich.
}\\
Email: spohn@ma.tum.de    
\end{center}     
 \date{March, 2013}     
 \vspace{2cm}
{\bf Abstract} 
We consider the Hamiltonian system of scalar wave field  and a single nonrelativistic 
particle coupled in a translation invariant manner. The particle is also subject 
to a confining external potential. The stationary solutions of the system are a Coulomb 
type wave field centered at those particle positions for which the external force vanishes.     
We prove that solutions of finite energy converge, in suitable local energy seminorms, 
to the set ${\cal S}$ of all stationary states in the long time limit $t\to\pm\infty$.
Further we show that
the rate of relaxation to a stable stationary state is determined by spatial decay 
of initial data. The convergence is followed by the radiation of the dispersion wave 
which is a solution to the free wave equation.

Similar relaxation has been proved previously for the case of relativistic particle 
when the speed of the particle is less than the speed of light. Now we extend these 
results to nonrelativistic particle with arbitrary superlight velocity.
However, we restrict ourselves by the plane particle trajectories in $\R^3$.
The extension to general case remains an open problem.
\end{titlepage}
\setcounter{equation}{0}     
\section{Introduction}
We consider the Hamiltonian system  of a real scalar field  $\varphi(x)$ on $\R^3$, 
and an extended nonrelativistic particle with the center position $q\in\R^3$ and with 
the charge density $\rho(x-q)$.     
The field is governed by the wave equation with a source. 
The particle is subject to the wave field and also to an external potential $V$,     
which is confining in the sense of (\re{P}).
The interaction between the particle and the scalar field is local,     
translation invariant, and linear in the field. We study the long-time 
behavior of the coupled system. Our main results are the asymptotics 
\be\la{dq-lim}
\dot q(t)\to 0,\quad\quad \ddot q(t)\to 0,\quad t\to\pm\infty,
\ee
and the convergence of the field to the corresponding Coulombic potential.
Moreover, we establish the rate of the convergence in the case when  $q_\pm$ 
is a nondegenerate local minimum of the  potential $V$.

Let $\pi(x)$ be the canonically conjugate field to $\varphi(x)$ and let $p$ be the momentum 
of the particle. The Hamiltonian (energy functional) reads then
\be\la{WP1}
{\cal H}(\varphi,q,\pi,p)\equiv\fr 12 p^2+V(q)+
\fr 12\int (|\pi(x)|^2+|\nabla\varphi(x)|^2)\, dx+\int \varphi(x)\rho(x-q)\, dx.
\ee
Taking formally variational derivatives in (\re{WP1}), the coupled dynamics becomes
\beqn\la{WP2}
\left.
\ba{ll}
\dot \varphi(x,t)=\pi(x,t),&\dot\pi(x,t)=\Delta\varphi(x,t)-\rho(x-q(t)),\\
~\\
\dot q(t)= p(t),&\dot p(t)=-\nabla V(q(t))+\ds\int \varphi(x,t)\nabla\rho(x-q(t)) dx.
\ea
\right|
\eeqn
For smooth $\varphi(x)$ vanishing at infinity the Hamiltonian
can be rewritten as 
\be\la{WP1r}
{\cal H}(\varphi,q,\pi,p)\equiv\fr 12 p^2+V(q)+
\fr 12\int (|\pi(x)|^2+|\nabla[\varphi(x)-\De^{-1}\rho(x-q)]|^2)\, dx+
\fr 12\langle\rho,\De^{-1}\rho\rangle,
\ee
where 
\begin{equation}\la{WPHB}
\fr 12\langle\rho,\De^{-1}\rho\rangle= 
-\fr 1 {8\pi}\int\int\fr{\rho(x)\rho(y)}{|x-y|}\,dxdy\le 0.
\end{equation}
Thus the energy  (\re{WP1r}) is bounded from below if $|\langle\rho,\Delta^{-1}\rho\rangle|<\infty$ 
which provides a priori estimates for solutions to (\re{WP2}), and hence guarantees 
the existence of global solutions. Otherwise, the dynamics is not well defined. 
For example, $\langle\rho,\Delta^{-1}\rho\rangle=-\infty$ for the point particle with $\rho(x)=\de(x)$: 
\be\la{ufd}
\langle\de,\Delta^{-1}\de\rangle=-(2\pi)^{-3}\int\fr{1}{k^2}dk=-\infty.
\ee
This ``ultraviolet divergence''  was discovered first for the point particle 
in classical electrodynamics, where $-\langle\rho,\Delta^{-1}\rho\rangle$ is proportional 
to the energy of the particle in its own electrostatic field. 
Respectively, the infinite energy (\re{ufd}) for the point particle is not satisfactory 
since it also means its infinite mass. This  infinity inspired the introduction 
of the ``extended electron'' by Abraham \ci{Abr2}.
Our system (\re{WP2}) is a scalar analog of the Abraham electrodynamics
with the extended electron \ci{KS00, Sp04}.
\medskip

The stationary solutions for (\re{WP2}) are easily determined. Denote
\be\la{sq}
s_{q}(x) =-\int\fr{\rho(y-q)}{4\pi|y-x|}dy,\quad x, q\in\R^3.
\ee
Let $Z=\{q\in\R^3: \nabla V(q)=0\}$ be the set of critical points for
$V$. Then the set of all stationary states is given by
\be\la{WPss}
{\cal S}=\{(\varphi, \pi, q,  p)=(s_{q}, 0, q, 0):=S_{q}|~~q\in Z\}.
\ee
One natural goal is to investigate the domain of attraction for ${\cal S}$ 
and in particular to prove that each finite energy  solution of (\re{WP2}) 
converges to some
stationary states $S_{q_{\pm}}=(s_{q_{\pm}}, 0, q_{\pm}, 0)\in {\cal S}$ in
the long time limit $t\to\pm\infty$.

To state our main results we need some assumptions on $V$ and $\rho$. We assume that
\be\la{P}
V\in C^2(\R^3),
~~~~~~\lim_{|q|\to\infty} V(q)=\infty.
\ee
\be\la{C}
\rho\in C_0^\infty(\R^3),
 ~~~~~\rho(x)=0{\rm ~~for~~}|x|\geq R_\rho,~~~~\rho(x)=\rho_r(|x|).
\ee
Moreover, we suppose that the following Wiener condition holds:
\be\la{W}
 \hat\rho(k)\ne 0\mbox{ ~~~~~~for~~ }k\in\R^3.
\ee
It is an analogue of the Fermi Golden Rule:
the coupling term  $\rho(x-q)$ is not orthogonal to the eigenfunctions $e^{ikx}$ 
of the continuous spectrum of the linear part of the   equation (cf. \ci{Sig,SW3}).
As we will see, the Wiener condition (\re{W}) is very essential for our asymptotic analysis.
\medskip

For technical reasons, we restrict ourselves to the case when the particle moves in the plane,
i.e. we suppose that $q(t)=(q^1(t), q^2(t), q^3(t))\in R^3$ such that
\be\la{Q}
q^3(t)=0,~~~~~~t\in\R.
\ee
For example, this condition holds if initial fields $\varphi_0(x)=\varphi(x,0)$ and
$\pi_0(x)=\pi(x,0)$ are symmetric in $x^3$, and 
\be\la{qVc}
q^3(0)=p^3(0)=0~~{\rm and}~~\quad \pa_{x^3} V(x^1,x^2,0)=0,~~{\rm for }~~ (x^1,x^2)\in\R^2.
\ee

In the first part of the paper  we prove that the set $\cal S$ is an attracting set 
for each trajectory $Y(t)=(\varphi(t),\pi(t),q(t),p(t))$. Namely, we consider initial data 
$Y(0)=(\varphi_0,\pi_0,q_0,p_0)$
with 
\be\la{C1}
\varphi_0\in C^2(\R^3),\quad \pi_0\in C^1(\R^3)
\ee
 such that
\be\la{WP8}
|\nabla\varphi_0(x)|+|\pi_0(x)|+|x|(|\nabla\nabla\varphi_0(x)|+|\nabla\pi_0(x)|)
={\cal O}(|x|^{-\si}),\quad |x|\to\infty, ~~{\rm where}~~\si>3/2,
\ee
which guarantees the finiteness of the energy   (\re{WP1}).
First, we prove the relaxation (\re{dq-lim}).
Further, we prove the long-time attraction
\be\la{WP9n} 
Y(t)\to {\cal S},\quad t\longrightarrow \pm\infty.
\ee
where the convergence of the fields holds in local energy seminorms.
If additionally, the set ${\cal S}$ is discrete, then (\re{WP9n}) implies 
\be\la{WP10n}
Y(t)\to S_{q_\pm},\quad t\longrightarrow \pm\infty,
\ee
where the stationary states $S_{q_\pm}\in {\cal S}$ depend on the solution $Y(t)$ considered.
\medskip

In the second part of the paper we specify the rate of convergence in (\re{WP10n})
to a stationary state $S_{q_+}$ in the case 
 where $q_{*}\in Z$ is a non-degenerate
minimum of the potential, i.e.,
\be\la{d2V}
d^2V(q_+)>0.
\ee
where $d^2V(q_+)$ is the Hessian.
We suppose that the initial fields belong to the weighted  space $\Ho^1_{\al}\oplus L^2_{\al}$
with some $\al>1$ (see Definition \ref{Def1}).  Then  for any $\ve>0$
\be\la{qqi}
\dot q(t)={\cal O}(|t|^{-\al+\ve}), ~~ q(t)=q_++{\cal O}(|t|^{-\al+\ve}),~~
\Vert (\varphi(t), \pi(t))-(s_{q_+},0)\Vert_{\Ho^1_{-\al}\oplus L^2_{-\al}}={\cal O}(t^{-\al+\ve}),\quad t\to\infty. 
\ee 
Moreover,  in this case the scattering asymptotics  hold,
\be\la{Si} 
(\varphi(x,t), \pi(x,t))= (s_{q_+}, 0) +W(t)\Phi_+ +r(x,t). 
\ee 
Here $W(t)$ is the dynamical group of the free wave equation, $\Phi\in\Ho^1\oplus L^2$ is the corresponding asymptotic state, 
and 
\be\la{rmi} 
\Vert r(t)\Vert_{\Ho^1\oplus L^2}={\cal O}(|t|^{-\al+1+\ve}),\quad t\to\infty. 
\ee 
The investigation is inspired by fundamental
 problems of the field theory and quantum mechanics.
Namely,  the relaxation 
of the acceleration (\re{dq-lim}) is known as {\it radiation damping} in classical 
electrodynamics since Lorentz and Abraham  \ci{Abr2}, however it was proved for the first time 
in \ci{KS00,KSK} for the case of relativistic particle with $\dot q=p/\sqrt{p^2+1}$.
Second, the asymptotics (\re{WP10n}) give a dynamical model of Bohr's transitions
to quantum stationary states, see  the details in \ci{K2013,K2016}.

Our extension to the nonrelativistic particle is not straightforward and important 
in connection with the Cherenkov radiation.
The main difficulty is due to the singular nature of the radiation for $|\dot q(t)|\ge 1$.

Traditionally  the classical 
Larmor and Li\'enard formulas
\ci[(14.22)]{Jackson1962} and \ci[(14.24)]{Jackson1962}
are accepted for
the power of radiation of a point particle.
These formulas contain the factor $(1-\beta  \cdot \om)^{-3}$ (cf. our formula (\re{Larmor}))
where $\beta=v/c$ and $\om$ is the direction of the radiation.
Here $v=\dot q(\tau)$ is the particle velocity at the ``retarded time'' $\tau$
and  $c$ is the  propagation speed of the wave field in the dispersive medium.
These formulas are deduced from the Li\'enard-Wiechert expressions for the retarded potentials
neglecting the initial fields. Moreover, 
these formulas neglect the back fieldreaction
though it should be the key reason for the relaxation. The main problem is that this back 
field-reaction is infinite for the point particles. In (\re{WP2}) we set $c=1$.
Generally $c$ is less than the speed of light in vacuum, so the particle velocities $\dot q(t)>1$
are possible. Then the factor $(1-\beta  \cdot \om)^{-3}$ in the Larmor formula becomes infinite 
for some directions $\om$.

A rigorous meaning to these calculations for relativistic particle has been suggested first in 
\ci{KS00,KSK} for the Abraham model of the "extended electron" under the Wiener condition (\re{W}). 
The survey can be found in \ci{Sp04}.

For the nonrelativistic Abraham type model (\re{WP2})  with the ``extended electron'' the radiation 
remains finite due to the smoothing by the coupling function $\rho$. Nevertheless, 
the case $|\dot q(t)|>1$ rises many open questions.
\medskip

Our main novelties in present paper are the following.
\medskip\\
I. Global attraction of finite energy solutions 
to stationary states for the case of  
nonrelativistic  particle.
\medskip\\
II.  Asymptotics (\re{qqi})--(\re{rmi})
 in the weighted Sobolev norms for the case
of nonrelativistic  particle.
\medskip

Let us comment on previous results in these directions. 
The global attractions (\re{WP9n}) 
and  (\re{WP10n}) were proved in \ci{KSK, KS00} for the system of type 
(\re{WP2}) with relativistic particle
and for the  similar Maxwell-Lorentz system.
In \ci{KS1998} the global attraction to solitons was proved for the system (\re{WP2}) without 
external potential  under the Wiener condition (\re{W}). In \ci{IKM2004}
this result was extended to similar Maxwell-Lorentz system.
In \ci{IKS2004a}--\ci{IKS2004b} the global attraction to solitons is proved for  the system 
(\re{WP2}) and similar systems with the Klein-Gordon and Maxwell equations with
small $\rho$.
In \ci{K2003}--\ci{KK2010b} the global attraction to solitary waves
is proved for the Klein--Gordon and Dirac equations coupled to $U(1)$-invariant nonlinear oscillators. 

The asymptotics of type (\re{qqi})--(\re{rmi}) were established in  \ci{KSK} for the case of 
relativistic  particle in local energy seminorms for initial fields with compact support.
In  \ci{KK16} we have proved the asymptotic stability of the stationary states  
for the system (\re{WP2}) in the weighted Sobolev norms. 

In a series of papers, Egli, Fr\"ohlich, Gang, Sigal, and Soffer  have established the
convergence to a~soliton  for the system of type (\re{WP2}) with the Schr\"odinger equation 
instead of the wave equation. The main result of \cite{FG2014} is the long time convergence 
to a soliton with a subsonic speed for initial solitons with supersonic speeds.
The convergence is considered as a reason for the Cherenkov radiation, 
see \cite{FG2014} and the references therein. 

The asymptotics of type (\re{Si}) were proved by Soffer and Weinstein 
for nonlinear Schr\"odinger equations with a potential \ci{SW1990,SW1992},
and for translation invariant   nonlinear Schr\"odinger equations by Buslaev,
Perelman and Sulem \ci{BP1993,BP1995,BS2003}.
\medskip

Now let us comment on our methods. For the  proof of (\re{dq-lim}) we estimate 
the energy dissipation by decomposing $\varphi$ into a near and far field. 
Energy is radiated in the far field. Since the Hamiltonian is bounded from below, such radiation
cannot go on forever and a certain "energy radiation functional" has to be bounded.
This radiation functional can be written as a convolution. By a Wiener  Tauberian Theorem, 
using (\re{W}), we conclude (\re{dq-lim}) for $\ddot q$. Therefore (\re{dq-lim}) also 
holds for $\dot q$ since $|q(t)|$ is bounded by some $\ov q_0<\infty$ due to (\re{P}).
Finally, we deduce (\re{WP9n}) and (\re{WP10n}) from (\re{dq-lim}) and 
integral representations for the fields.
This strategy is close to \ci{KSK,KS00,Sp04}, however, the singularity of the radiation 
at $|\dot q(t)|\ge 1$ requires suitable modifications in application 
of the Wiener Tauberian Theorem. We suggest the modification for  the plane particle 
trajectories (\re{Q}). The extension to general case remains an open problem.

We prove the asymptotics  (\re{qqi})--(\re{rmi})  by 
a development of the methods of \ci{KSK} and 
controlling the nonlinear part 
of (\re{WP2}) by the dispersion decay for the linearized equation which we established in
 \ci{KK16}.
Let us emphasize however, that the  
asymptotics (\re{qqi})--(\re{rmi}) are  quite different 
from the asymptotic stability proved in  \ci{KK16}.

\medskip
 
The plan of our paper is as follows.
In \S \ref{exdyn} we introduce appropriate functional spaces  and  
formulate our main results.
In \S \ref{LW-as} we refine known results on the long range asymptotics  of the 
 Li\'enard-Wiehert potentials. In \S\ref{scat-sec} we calculate the energy radiation integral.
 We use this formula in \S\ref{relax-sec} to prove the velocity relaxation.
 In \S\ref{Tr-sec} we prove the attraction to stationary states.
 In \S \ref{lin-sec}   we consider  the linearization at stationary state.
 In \S\ref{nH-sec} we prove a version of strong  Huygens principle for nonlinear system (\re{WP2}).
 In  \S\S \ref{rc-sec}--\ref{as-sec} we deduce the asymptotics  (\re{qqi})--(\re{rmi}).
 
 \setcounter{equation}{0}
 \section{Existence of dynamics and main results}\label{exdyn}
 We consider the Cauchy problem for the Hamiltonian system (\re{WP2})
which can be written as
\be\la{WP2.1}
\dot Y(t)=F(Y(t)),~~t\in\R,~~Y(0)=Y_0.
\ee
 Here $Y(t)=(\varphi(t), \pi(t), q(t), p(t))$, $Y_0=(\varphi_0, \pi_0, q_0, p_0)$, and
all derivatives are understood in the sense of distributions.

Now we  introduce a suitable phase space. Let $L^2$ be the real Hilbert space $L^2({\R}^3)$ with scalar product
$\langle\cdot,\cdot\rangle$ and norm $\Vert\cdot\Vert_{L^2}$,  and let $H^1$ denote the Sobolev space
$H^1=\{\psi\in L^2:\,|\nabla\psi|\in L^2\}$ with the norm 
$\Vert\psi\Vert_{H^1}=\Vert\nabla\psi\Vert_{L^2}+\Vert\psi\Vert_{L^2}$.
For $\alpha\in\R$ 
let us define
by $L^2_{\alpha}$ 
  the weighted Sobolev spaces $L^2_{\alpha}$ 
with the norms $\Vert\psi\Vert_{L^2_{\alpha}}:=\Vert(1+|x|)^{\alpha}\psi\Vert_{L^2}$.

Denote by $\Ho^1$ the completion of real space $C_0^\infty(\R^3)$ with the 
norm $\Vert\na\varphi(x)\Vert_{L^2}$.
Equivalently, using Sobolev's embedding theorem,
$\Ho^1=\{\varphi(x)\in L^6(\R^3):~~|\nabla\varphi(x)|\in L^2\}$.
Denote by $\Ho^1_\alpha$ the completion of real space 
 $C_0^\infty(\R^3)$ with the 
norm $\Vert(1+|x|)^{\alpha}\na\varphi(x)\Vert_{L^2}$.

For any $R>0$ denote by $\Vert\varphi\Vert_{L^2(B_R)}$ the norm in $L^2(B_R)$,
where $B_R=\{x\in\R^3: |x|\le R\}$. Then the seminorms
$\Vert\varphi\Vert_{H^1(B_R)}=\Vert\nabla\varphi\Vert_{L^2(B_R)} +\Vert\varphi\Vert_{L^2(B_R)}$ 
are continuous on $\Ho^1$.

\begin{definition}\la{Def1}
 i) The phase space $\cE$ is the real Hilbert space $\Ho^1\oplus L^2\oplus {\R}^3\oplus {\R}^3$ of states
$Y=(\psi ,\pi ,q,p)$ with the finite norm
$$
\Vert Y\Vert_{\cE}=\Vert \na\psi \Vert_{L^2} +\Vert\pi \Vert_{L^2}+|q|+|p|.
$$
ii) ${\cal E}_F$ is the space ${\cal E}$ endowed with the Fr\'echet topology defined by the local energy seminorms
\be\la{WP6}
{\Vert Y\Vert}_R=\Vert\varphi\Vert_{H^1(B_R)}+ \Vert\pi\Vert_{L^2(B_R)} + |q|+ |p|, \quad\forall R>0.
\ee
iii) ${\cal E}_{\alpha}$ with $\alpha\in\R$ is the Hilbert space $\Ho^1_{\alpha}\oplus L^2_{\alpha}\oplus {\R}^3\oplus {\R}^3$
with the norm
\be\la{alfa}
\Vert Y\Vert_{\alpha}=\Vert \,Y\Vert_{{\cal E}_{\alpha}}
=\Vert \nabla\psi \Vert_{L^2_\alpha} +\Vert\pi \Vert_{L^2_\alpha}+|q|+|p|.
\ee
iv) ${\cal F}_{\alpha}$ is the space $\Ho^1_{\alpha}\oplus L^2_{\alpha}$ of fields $F =(\psi ,\pi )$ with the finite norm
\be\la{Falfa}
\Vert \,F\Vert_{\alpha}=\Vert \,F\Vert_{{\cal F}_{\alpha}}=
\Vert \nabla\psi \Vert_{L^2_\alpha} +\Vert\pi \Vert_{L^2_\alpha}.
\ee
\end{definition}
Note that we use the same notation for the norms in the space $\mathcal{F}_{\alpha}$ as in the space
$\mathcal{E}_{\alpha}$ defined in (\ref{alfa}). We hope it will not create misunderstandings since it will be always clear 
from the context if we deal with fields only, and therefore with the space $\mathcal{F}_{\alpha}$, 
or with fields-particles, and therefore with elements of the space $\mathcal{E}_{\alpha}$.

Note that both spaces ${\cal E}_F$ and ${\cal E}$ are metrisable, $\Ho^1$ is not contained in $L^2$ and for instance 
$\Vert S_q\Vert_{L^2}=\infty$. On the other hand, $S_q\in{\cal  E}$. Therefore, ${\cal E}$ is the space of finite energy states.
The Hamiltonian functional  (\re{WP1r}) is continuous on the space ${\cal E}$ and 
is bounded from below. 
In the point charge limit the lower bound tends to $-\infty$ by (\re{ufd}).
 \begin{lemma}\la{WPexistence} (see \ci [Lemma 2.1]{KSK})
Let conditions (\re{P}) and (\re{C}) hold. Then\\
 (i) For every $Y_0\in {\cal E}$ the Cauchy problem (\re{WP2.1}) has a unique solution $Y(t)\in C(\R, {\cal E})$.\\
 (ii) For every $t\in\R$ the map $U(t):Y_0\mapsto Y(t)$ is continuous both on ${\cal E}$ and on ${\cal E}_F$.\\
 (iii) The energy is conserved, i.e.
 \be\la{WPEC}
  {\cal H}(Y(t))={\cal H}(Y_0)~~~for~~t\in\R.
 \ee 
  (iv) The following a priori estimates hold
\be\la{aps}
\Vert Y(t)\Vert_{\cE} \le C(Y_0),\qquad t\in\R.
\ee 
(v) The time derivatives $q^{(k)}(t)$, $k=0,1,2,3$, are uniformly bounded, i.e.~there are constants $\ov q_k>0$, 
depending only on the initial data, such that
 \be\la{WP2.4} 
   |q^{(k)}(t)|\le \ov q_k~~~for~~t\in\R.
 \ee
 \end{lemma}
 Our first main result is the following theorem.
 \begin{theorem}\la{t1}
 Let conditions (\re{P})--(\re{Q}) and  (\re{C1})-(\re{WP8}) hold. Then for the corresponding  solution 
 $Y(t)\in {\cal E}$ to the Cauchy problem (\re{WP2.1})\\
 i) 
 The attraction holds
\be\la{WP9} 
Y(t)\tocEF {\cal S},\quad t\longrightarrow \pm\infty.
\ee
ii) If additionally, the set ${\cal S}$ is discrete, then (\re{WP9}) implies similar  convergence
\be\la{WP10}
Y(t)\tocEF S_\pm,\quad t\longrightarrow \pm\infty.
\ee
\end{theorem}

 Our second main result  refine the asymptotics 
(\re{WP9})-- (\re{WP10}) for initial fields from the Sobolev weighted spaces.

 \begin{theorem}\la{WPC}
Let conditions (\re{P})--(\re{W}) hold, and let $Y(t)\in C(\R,{\cal E})$ be a solution to the Cauchy 
 problem  (\re{WP2.1}) with  $Y_0\in{\cal E}_\al$, where $\al >1$. 
Suppose that
\be\la{WPconv1}
Y(t)\tocEF S_{q_+},\quad \quad t\to \infty
\ee
where  the limit point $q_+\in Z$ satisfies (\re{d2V}).
Then \\
i) For every $\ve>0$
\be\la{WPEE1}
 {\|Y(t)-S_{q_+}\|}_{-\al}={\cal O}(t^{-\al+\ve}),\quad t\to\infty.
\ee
ii) For every $\ve>0$ the scattering asymptotics  hold,
\be\la{S} 
(\varphi(x,t), \pi(x,t))= (s_{q_+}, 0) +W(t)\Phi_+ +r(x,t), 
\ee 
where $\Phi_+\in\Ho^1\oplus L^2$, and
\be\la{rm} 
\Vert r(t)\Vert_{\Ho^1\oplus L^2}={\cal O}(|t|^{-\al+1+\ve}),\quad t\to \infty.
\ee 
\end{theorem}

\setcounter{equation}{0}
\section{Li\'enard-Wiechert asymptotics}\label{LW-as}
The solution 
to the non-homogeneous wave equation from the system 
(\re{WP2})  is the sum of 
two terms. The first is 
the retarded Li\'enard-Wiechert potential (\re{WPretp})
which is the 
solution to the non-homogeneous wave equation with zero initial data.
The second term is 
the solution to the homogeneous equation 
with the initial data of the total field. This term is given by the Kirchhoff formula 
(\re{kf}).

The second term does not does not affect the long-time asymptotics of the solution
due to the strong Huygens principle. Thus, exactly 
the retarded Li\'enard-Wiechert potential is responsible for the long-time asymptotics. 

In this section 
we refine the  results \ci{KSK,KS00} on the long time and long range asymptotics 
of the  Li\'enard-Wiechert potentials
\be\la{WPretp}
  \varphi_r(x,t)=-\fr 1{4\pi}\int\fr{dy~\theta(t-|x-y|)}{|x-y|}\rho(y-q(t-|x-y|)),
\qquad \pi_r(x,t)=\dot\varphi_r(x,t).
\ee
These asymptotics will play the key role in subsequent calculation
of the energy radiation which is used
 in the proof of the relaxation (\re{dq-lim}).
Furthermore, we estimate the energy radiation corresponding  to
 the Kirchhoff integral (\re{kf}). 
\medskip

First, we prove asymptotics  of the retarded potentials  
in the wave zone $|x|\sim t\to\infty$. These asymptotics and their proofs
are similar to that of  Lemma 3.2 of \ci{KS00}.
\begin{lemma}\la{WPKR} Let conditions (\re{P}) and  (\re{C}) hold. Then
there exists  $T_r>0$ such that
the following asymptotics hold uniformly in $t\in[T_r,T]$
for every fixed $T>T_r$,
\beqn
\pi_r(x,|x|+t)&=&\,\,\,\,\ov \pi(\om(x),t)|x|^{-1}+{\cal O}(|x|^{-2}),\la{WP3.7}\\
\nabla\varphi_r(x,|x|+t)&=&-\om(x)\ov \pi(\om(x),t)|x|^{-1}+{\cal O}(|x|^{-2})\la{WP3.7'}
\eeqn 
as $|x|\to\infty$ with a function $\ov \pi(\om,t)$. Here $\om(x)=x/|x|$.
\end{lemma}
\begin {proof}
The integrand of (\re{WPretp}) vanishes for $|y|>T_r:=\ov q_0+R_\rho$. Then for $t-|x|>T_r$ one has
$$
|x-y|\le |x|+|y|\le t-T_r+T_r\le t,
$$
and hence (\re{WPretp}) implies that
\be\la{WP3.12}
\pi _r(x,t)=-\int dy~\fr 1 {4\pi |x-y|}\nabla\rho(y-q(\tau))\cdot\dot q(\tau),
\ee
where $\tau=t-|x-y|$. Similarly, for  $t-|x|>T_r$ 
\beqn\la{WP3.15}
\nabla\varphi_r(x,t)&=&\int dy~\fr 1 {4\pi |x-y|} n\nabla\rho(y-q(\tau))\cdot\dot q(\tau)
+{\cal O}(|x|^{-2})\nonumber\\
&=&-\om(x) \pi_r(x,t)+{\cal O}(|x|^{-2}),
\eeqn
since $\ds n=\fr{x-y}{|x-y|}=\om(x)+{\cal O}(|x|^{-1})$ for bounded $y$.
Now we substitute $|x|+t$ instead of $t$ in  representations (\re{WP3.12}), (\re{WP3.15})
to get  asymptotics (\re{WP3.7}), (\re{WP3.7'}) for $t>T_r$. Then $\tau$ becomes 
\be\la{WP3.16}
\tau=|x|+t-|x-y|=t+\omega(x) \cdot y+{\cal O}(|x|^{-1})=\ov\tau+{\cal O}(|x|^{-1}),\qquad \ov\tau=t+\om\cdot y,
\ee
since
$$
|x|-|x-y|=|x|-\sqrt{|x|^2-2x\cdot y+|y|^2}\sim |x|\Big(\fr {x\cdot y}{|x|^2}+\fr{|y|^2}{2|x|^2}\Big)
=\omega(x) \cdot y+{\cal O}(|x|^{-1}).
$$
 Hence (\re{WP3.12})  implies (\re{WP3.7}) with
\be\la{WP3.17}
\ov \pi(\om,t)=-\fr 1 {4\pi}\int dy~\nabla\rho(y-q(\ov\tau))\cdot\dot q(\ov\tau).
\ee
Then (\re{WP3.15}) gives (\re{WP3.7'}) immediately.
\end{proof}
Note that asymptotics (\re{WP3.7}) - (\re{WP3.7'}) hold  without condition (\re{Q}). 
However, this condition  allows us to represent (\re{WP3.17}) in a more efficient way for
$\om$ close to $(0,0,\pm 1)$, see next lemma. Namely,  let us denote
\beqn\la{theta}
\Theta=\left\{\ba{rl}
0,&\ov q_1<1\\
\ve+\sqrt{1-(\ov q_1)^{-2}},&\ov q_1\ge 1\\
\ea\right.
\eeqn
with an arbitrary small $0<\ve<1-\sqrt{1-(\ov q_1)^{-2}}$.
Then for $\om=(\om^1,\om^2, \om^3)$ with $|\om^3|\ge\Theta$  we obtain
\be\la{le1}
|\om\cdot\dot q|=|\dot q| |\cos(\om,\dot q)|\le\ov q_1\sqrt{1-(\om^3)^2}\le\ov q_1\sqrt{1-\Theta^2}<1.
\ee 
\begin{lemma}\la{New}
Let  conditions  (\re{P}), (\re{C}) and (\re{Q}) hold. Then for any $\om$ with $|\om^3|\ge\Theta$  one has
\be\la{WP3.0}
\ov \pi(\om,t)=\fr 1 {4\pi}\int dy~\rho(y-q(\ov\tau))\frac{\om\cdot\ddot q(\tau)}{(1-\om\cdot\dot q(\tau))^2}
\ee
\end{lemma}
\begin{proof}
We observe that
$$
\nabla_y\rho(y-q(\ov\tau))\cdot\dot q(\ov\tau)=
 \nabla\rho(y-q(\ov\tau))\cdot\dot q(\ov\tau)~(1-\om\cdot\dot q(\ov\tau)).
$$
Then  (\re{le1}) implies
\beqn\la{WP3.18}
\int dy~\nabla\rho(y-q(\ov\tau))\cdot\dot q(\ov\tau)&=&\,\,\,\,\int dy~
\nabla_y\rho(y-q(\ov\tau))\cdot\dot q(\ov\tau)\fr 1{1-\om\cdot\dot q(\ov\tau)}\nonumber\\
&=&-\int dy~\rho(y-q(\ov\tau))\sum\limits_{j=1}^2\fr{\pa}{\pa y^j}\fr{\dot q^j(\ov\tau)}{1-\om\cdot\dot q(\ov\tau)}.
\eeqn
Differentiating, we get
\be\la{WP3.19}
\sum\limits_{j=1}^2\fr{\pa}{\pa y^j}\fr{\dot q^j}{1-\om\cdot\dot q}=\fr{\om\cdot\ddot q}{(1-\om\cdot\dot q)^2}.
\ee
Then (\re{WP3.17}) agrees evidently with (\re{WP3.0}).
\end{proof}
Denote $(\varphi_K(t),\pi_K(t)):=W(t)[(\varphi_0,\pi_0)]$, where 
 $\varphi_K(x,t)$ is  the Kirchhoff integral
\be\la{kf}
 \varphi_K(x,t)=\fr 1{4\pi t}\int_{S_t(x)}~d^2y~\pi_0(y)+
 \fr \pa {\pa t}~\Bigg(~\fr 1{4\pi t}\int_{S_t(x)}~d^2y~\varphi_0(y)~\Bigg),
\ee
and $\pi_K(x,t)=\dot\varphi_K(x,t)$. Here $S_t(x)$ denotes the sphere $\{y:~|y-x|=t\}$ and $d^2y$ is the corresponding 
surface area element. Below we will use the following lemma:
\begin{lemma}\la{WPI00}
Let $(\varphi_0,\pi_0)$ satisfies (\re {C1}) and (\re {WP8}). Then
there exist $I_0<\infty$  such that for every $R>0$ and every $T>T_0\ge0$
\be\la{WP3.10}
 \int_{R+T_0}^{R+T}dt \int_{S_R}d^2x~\Big(|\pi_K(x,t)|^2+|\nabla\varphi_K(x,t)|^2\Big) \leq I_0.
\ee
Here and below $S_R=S_R(0)$.
\end{lemma}
\begin{proof}
Formula (\re{kf}) implies 
$$
\varphi_K(x,t)=\fr{t}{4\pi}\int_{S_1}d^2z~\pi_0(x+tz)+\fr{1}{4\pi}\int_{S_1}d^2z~\varphi_0(x+tz)
+\fr{t}{4\pi}\int_{S_1}d^2z~\na\varphi_0(x+tz)\cdot z.
$$
Therefore
$$
\nabla\varphi_K(x,t)=\fr{t}{4\pi}\int_{S_1}\!\! d^2z~\na\pi_0(x+tz)+
\fr{1}{4\pi}\int_{S_1}\!\! d^2z~\na\varphi_0(x+tz)+
\fr{t}{4\pi}\int_{S_1}\!\! d^2z~\na_x(\na\varphi_0(x+tz)\cdot z).
$$
Here all derivatives are understood in the classical sense.
A similar representation holds for $\pi_K(x,t)$.
Hence, taking into account the assumption (\re{WP8}), we obtain 
\be\la{WPbdt}
|\pi_K(x,t)|,~|\nabla\varphi_K(x,t)|\le C \sum_{s=0}^1 t^{s}\int_{S_1}d^2z~|x+tz|^{-\sigma-s},\quad\si>3/2.
\ee
Further, for $\sigma\not =2$ we have
$$
\int_{S_1}d^2z~|x+tz|^{-\sigma-s}=\fr{2\pi}{(\si+s-2)|x|t}\Big((t-|x|)^{2-\si-s}-(t+|x|)^{2-\si-s}\Big),
\quad s=0,1.
$$
Therefore, 
\beqn\nonumber
 \int\limits_{ R+T_0}^{R+T}dt \int\limits_{S_R}d^2x\Big(|\pi_K(x,t)|^2+|\nabla\varphi_K(x,t)|^2\Big) 
\!\!\!&\le&\!\!\! C 
 \int\limits_{R+T_0}^{R+T} \Big[\fr{(t+R)^{4-2\si}+(t-R)^{4-2\si}}{t^2}+(t-R)^{2-2\sigma}\Big]dt\\
\nonumber
\!\!\! &\le& \!\!\! C_1\int\limits_{R+T_0}^{R+T}dt \Big[\Big(1+\fr Rt\Big)^{2}+\Big(1-\fr Rt\Big)^{2}+1\Big](t-R)^{2-2\sigma}
<\infty.
 \eeqn
\end{proof}
 \setcounter{equation}{0}
\section{Scattering of energy to infinity}\label{scat-sec}
In this section we establish a lower bound on the total energy radiated to infinity in terms of a "radiation integral".
Since the energy is bounded a priori, this integral has to be finite, which is then our main input
for proving Theorem \re{t1}.
\begin{pro}\la{WPaussen}
Let conditions (\re{P}), (\re{C}), (\re {C1}), (\re {WP8}) hold, and  let
$Y(t)=(\varphi(t),\pi(t), q(t), p(t))\in C(\R, {\cal E})$ be the solution to (\re{WP1}) with initial data $Y(0)=(\varphi_0,\pi_0, q_0, p_0)$.
 Then
\be\la{WP3.1}
\int_{0}^\infty dt\int_{S_1} d^2\om|\ov \pi(\om,t)|^2<\infty.
\ee
\end{pro}
 \begin{proof}
{\it Step i)}. The energy ${\cal H}_R(t)$ at time $t\in\R$ in the ball $B_R$ with 
a radius $R>\ov q_0+R_\rho$ is defined by
\be\la{WP3.2}
{\cal H}_R(t)=\fr 12\int_{B_R}dx ~\Big(|\pi(x,t)|^2+|\nabla\varphi(x,t)|^2\Big)
+ \fr 12 p^2(t)+V(q(t))+\int_{\R^3} dx~ \varphi(x,t)\rho(x-q(t))\,.
\ee
Let us fix a $R>0$ and consider a total radiated energy ${\cal H}_R(R+T_0)-{\cal H}_R(R+T)$
from the ball $B_R$ during the time interval $[R+T_0, R+T]$, where $T>T_0\ge 0$.
This quantity is bounded a priori, because ${\cal H}_R(R+T_0)$ and 
 ${\cal H}_R(R+T)$
are bounded by (\re{aps}). 
Hence,
\be\la{WP3.4}
{\cal H}_R(R+T_0)-{\cal H}_R(R+T)\le I<\infty,
\ee
where  $I$ does not depend on $T_0$, $T$  and $R$.
\medskip\\
{\it Step ii)}.
Note that  the function $\varphi(x,t)=\varphi_r(x,t)+\varphi_K(x,t)$ is $C^1$ 
differentiable
in the region $t>|x|$ by (\re{WP8}), (\re{WPretp}) and (\re{kf}). Hence, 
differentiating  
(\re{WP3.2}) in $t$ and integrating by parts,
we get
\be\la{WP3.5}
\fr d{dt}{\cal H}_R(t)=\int_{S_R}d^2 x ~\omega (x)\cdot \pi(x,t)\nabla\varphi(x,t),
\quad t>R.
\ee
Now (\re{WP3.5}) and (\re{WP3.4}) imply
\be\la{WP3.6}
-\int_{R+T_0}^{R+T}~dt~\int_{S_R}d^2 x ~\omega (x)\cdot \pi(x,t)\nabla\varphi(x,t)\le I.
\ee
{\it Step iii)}. 
Let us show that \ref{WP3.6}  leads to (\re{WP3.1}) in the limits $R\to\infty$
 and then $T\to\infty$.
Indeed, 
substituting
\be\la{WPRdec} 
\pi=\pi_r+\pi_K,\quad \varphi=\varphi_r+\varphi_K 
\ee
into (\re{WP3.6}), we obtain
\be\la{WP3.9}
-\int_{R+T_0}^{R+T}~dt~\int_{S_R}d^2 x ~\omega (x)\cdot
(\pi_r \nabla\varphi_r+\pi_K \nabla\varphi_r+\pi_r \nabla\varphi_K
+\pi_K \nabla\varphi_K)\le I.
\ee
Then Lemmas \re{WPKR} and \re{WPI00}  imply for every fixed $T>T_0:=T_r$,
\be\la{WP3.10'}
\int_{T_r}^T~dt~\int_{S_1}d^2 \om ~|\ov \pi(\om,t)|^2\le I_1+T{\cal O}(R^{-1}),
\ee
where $I_1<\infty$ does not depend on $T$ and $R$. This follows by the Cauchy-Schwarz
 inequality.
Taking the limit $R\to\infty$ and then $T\to\infty$ we obtain (\re{WP3.1}).
 \end{proof}
 \setcounter{equation}{0}
 \section{Relaxation of the particle acceleration and velocity}\label{relax-sec}
In this section we  deduce  the relaxation  $\dot q(t)\to 0$, $\ddot q(t)\to 0$ as $t\to\infty$ using 
Proposition \re{WPaussen}. First, the function
 \be\la{WP4.2}
 \ov\pi(\omega,t)= \fr 1{4\pi}\int dy\,
\rho(y-q(t+\omega\cdot y))\frac{\omega\cdot\ddot q(t+\omega\cdot y)}
 {{(1-\omega\cdot\dot q(t+\omega\cdot y))}^2}
 \ee
is globally Lipschitz-continuous in $\omega$ and $t$ for $|\om^3|\ge\Theta$ due to  
(\re{le1}) and
the bounds  (\re{WP2.4})
with $k=2,3$.  Hence, Proposition \re{WPaussen} implies that
 \be\la{WP4.3}
 \lim_{t\to\infty} \ov\pi(\omega,t)=0
 \ee
uniformly in $\omega\in \Om(\Theta):=\{
\om\in S_1\!\!:|\om^3|\ge\Theta\}$.
Denote $r(t)=\omega\cdot q(t)\in\R$, $s=\omega\cdot y$, and
 $\rho_a(q^3)=\ds\int dq^1dq^2\rho(q^1,q^2,q^3)$, 
 and decompose  in (\re{WP4.2})  the $y$-integration along and transversal to $\omega$. Then
\begin{eqnarray}\la{WP4.4}
\ov\pi(\omega,t) & = & \int ds \,\rho_a(s-r(t+s))\,\frac{\ddot r(t+s)}
 {{(1-\dot r(t+s))}^2}\nonumber\\ & = &\int d\tau\, \rho_a(t-(\tau-r(\tau)))\,
 \frac{\ddot r(\tau)}{{(1-\dot r(\tau))}^2}= 
\int d\theta\, \rho_a(t-\theta) g_\omega(\theta)=\rho_a*g_\omega(t).
\end{eqnarray}
Here we substituted $\theta=\theta(\tau)=\tau-r(\tau)$, which is a nondegenerate diffeomorphism
since $\dot r\leq \ov r<1$ due to (\re{le1}), and we set
\be\la{Larmor}
 g_{\omega}(\theta)=\fr {\ddot r(\tau(\theta))}{(1-\dot r(\tau(\theta)))^3}
,\qquad\omega\in \Om(\Theta).
\ee
Now we extend $q(t)$ smoothly to zero for $t<0$. Then $\tilde\rho*g_\omega\,(t)$
 is defined for all $t$ and agrees 
with $\ov\pi(\omega,t)$ for sufficiently large $t$. Hence (\re{WP4.3}) 
reads as a convolution limit
\be\la{WP4.6}
 \lim_{t\to\infty} \rho_a*g_{\omega}(t)=0,\qquad\omega\in \Om(\Theta).
\ee
Now note that  (\re{WP2.4}) with $k=2,3$ imply that $g^\prime_\omega(\theta)$  is bounded.
Hence (\re{WP4.6}) and (\re{W}) imply by Pitt's extension to Wiener's Tauberian Theorem, cf.~\cite[Thm.~9.7(b)]{Ru},
\be\la{WP4.7}
 \lim_{\theta\to\infty} g_{\omega}(\theta)=0,\qquad\omega\in \Om(\Theta).
\ee
\begin{lemma}\la{WPWi}
 Let conditions (\re{P})--(\re{Q}) and  (\re{C1})-(\re{WP8}) hold, and let    
 $Y(t)\in {\cal E}$ be the corresponding  solution to the Cauchy problem (\re{WP2.1}). 
Then
\be\la{WP4.7'}
\lim_{t\to\infty} \ddot q(t)=0.
\ee
\end{lemma}
\begin{proof}
The limit (\re {WP4.7}) holds for any $\omega\in S_1$ with $|\om^3|\ge\Theta$ (see (\re{theta})).
Moreover, $\theta(t)\to\infty$ as $t\to\infty$. 
Hence, $\ddot r(t)=\om\cdot\ddot q(t)\to 0$ as $t\to\infty$ for any 
$\omega\in \Om(\Theta)$.
\end{proof}
\brs
(i) For a point charge $\rho(x)=\delta(x)$
we have $\rho_a(s)=\de(s)$. Hence,
(\re{WP4.6}) implies (\re{WP4.7}) directly, without the application of the 
Wiener Tauberian Theorem.\\
(ii) Condition (\re{W}) is necessary for the implication (\re{WP4.7})$\Rightarrow$(\re{WP4.7'}). 
Indeed, if (\re{W}) is violated, then $\hat\rho_a(\xi)=0$ for some $\xi\in\R$, and with the choice
$g(\theta)=\exp(i\xi \theta)$ we have $\rho_a*g(t)=0$ whereas $g$ does not decay to zero.
 \ers
\begin{cor}
Let conditions of Lemma \re{WPWi}  hold.  Then
\be\la{WPrel}
    \lim_{t\to\infty}\dot q(t)=0.
\ee
\end{cor}
\begin{proof}
(\re{WP4.7'}) implies (\re{WPrel}) since $|q(t)|\le \ov q_0$ due to (\re{WP2.4}) with $k=0$.
\end{proof}
\setcounter{equation}{0}
\section{Transitions to stationary states}\label{Tr-sec}
Here we prove our main Theorem \re{t1}. First we show that the set
\be\la{cA}
{\cal A} = \{S_q:~q=(q^1,q^2,0)\in\R^3,~ |q|\le \ov q_0)\}
\ee 
is an attracting subset. It is compact in ${\cal E}_F$ since ${\cal A}$ is homeomorphic to a closed ball in $\R^3$.
\begin{lemma}\la{WPCAT}
 Let conditions of Theorem \ref{t1} hold. Then
\be\la{WPcA} 
Y(t)\tocEF {\cal A},\quad t\longrightarrow \pm\infty.
\ee
\end{lemma}
\begin{proof}
It suffices to verify that for every $R>0$ 
\begin{eqnarray}\la{WP5.1}
\Vert Y(t)- S_{q(t)} \Vert_R=\Vert\varphi(t)-s_{q(t)}\Vert_{H^1(B_R)}+\Vert\pi(t)\Vert_{L^2(B_R)}+|p(t)|
\to 0\,\,{\rm as}\,\,t\to\infty.
\end{eqnarray}
Let us estimate each term separately.
\medskip\\
 i) Convergence (\re{WPrel}) implies that $|p(t)|\to 0$ as $t\to\infty$.
 \medskip\\
 ii)  The  integral representation (\re{WP3.12}) implies that for $|x|<R$
 and $t>R+T_r$,  $T_r=\ov q_0+R_\rho$,  we have
\[
 |\pi_r(x,t)|\le~C\max\limits_{\tau\in[t-R-T_r, t]}|\dot q(\tau)| \int_{|y|<T_r} dy \,\fr 1{|x-y|}|\nabla\rho (y-q(t-|x-y|))|.
\]
 Here the integral is bounded uniformly in $t>R+T_r$ for $x\in B_R$, and therefore
 (\re{WPrel}) implies that $\Vert\pi_r(t)\Vert_{L^2(B_R)}\to 0$ as $t\to\infty$.
 Hence,  $\Vert\pi(t)\Vert_{L^2(B_R)}\to 0$ by (\re{WPRdec}) and (\re{WPbdt}).
 \medskip\\
 iii)
The integral representation (\re{WPretp}) implies for $t>R+T_r$ and $|x|<R$ that
$$
 \varphi_r(x,t)-s_{q(t)}(x) =-\int_{|y|<T_r} dy\fr 1{4\pi|x-y|}\Big(\rho(y-q(t-|x-y|))-\rho(y-q(t))\Big)\,.
$$
The difference $q(t-|x-y|)-q(t)$ may be written as an
 integral depending only on $\dot q(\tau)$ for $\tau\in [t-R-T_r,t]$, 
 which tends to zero as $t\to\infty$ uniformly in $x\in B_R$ due to (\re{WPrel}).
 Hence
 $\Vert\varphi_r(t)-\varphi_{q(t)}\Vert_{L^2(B_R)}\to 0$ as $t\to\infty$.
 Then  $\Vert\varphi(t)-\varphi_{q(t)}\Vert_{L^2(B_R)}\to 0$ by (\re{WPRdec})
and (\re{WPbdt}).
This proves the claim, since $\Vert\nabla(\varphi(t)- \varphi_{q(t)})\Vert_{L^2(B_R)}$ may be estimated
in a similar way.
 \end{proof}
Now we prove the convergences (\re{WP9}).
\begin{lemma}\la{WPCST}
Under conditions  of Theorem \ref{t1} the convergence holds
\be\la{WPcS} 
Y(t)\tocEF {\cal S},\quad t\longrightarrow \pm\infty.
\ee
\end{lemma}
\begin{proof}
Lemma \re{WPCAT}  implies that the orbit $O(Y):=\{Y(t):t\in\R\}$ is precompact 
in ${\cal E}_F$ since $\cA$ is the compact set 
in ${\cal E}_F$.
Let us denote by $\Om$ the set of all omega-limit points of 
the orbit in
 ${\cal E}_F$:
$\ov Y\in \Om$ means by definition that 
\be\la{oml}
 Y(t_k)\tocEF\ov Y,\qquad t_k\to\infty.
\ee
It suffices to prove that $\Om\subset\cS$, i.e. that any omega-limit point 
$\ov Y=S_{q_+}$ with some $q_+\in Z$.

First, Lemma \re{WPCAT} implies that $\ov Y\in {\cal A}$.
 Further, $\Om$ is invariant with respect to the dynamical group $U(t)$ with $t\in\R$
due to the continuity of $U(t)$ in ${\cal E}_F$. 
Hence, there exists a $C^2$-curve $t\mapsto Q(t)\in\R^3$ such that
$U(t)\ov Y=S_{Q(t)}$, according to Definition (\re{cA}).
However, for $S_{Q(t)}$ to be a solution of (\re{WP2}) we must have $\dot Q(t)\equiv 0$, 
and hence $Q(t)\equiv q_+\in Z$. Therefore, 
$\ov Y=S_{q_+}\in\cal S$.
\end{proof}

At last, we formalize the implication (\re{WP9}) $\Rightarrow$ (\re{WP10}) by the following definition.
Let ${\cal T}$ be a subset of a metrisable space ${\cal F}$.
\begin{definition}\la{dT}
${\cal T}$ is a trapping set in ${\cal F}$, if for every continuous curve $Y(t)\in C(\R, {\cal F})$ with a precompact 
orbit $O(Y)$ the convergence $Y(t)\tocF {\cal T}$ as $t\to\infty$ implies the convergence 
$Y(t)\tocF T$ as $t\to\infty$ to some point $T\in {\cal T}$.
\end{definition}
For example every discrete set in $\R^3$ is a trapping set in $\R^3$.
\begin{lemma}\la{WPA}
Let the conditions of of Theorem \ref{t1} hold  and let  $Z$ be a trapping set in $\R^3$. 
Then there exist stationary states $S^\pm  \in{\cal S} $ depending on $Y_0$ such that (\re{WP10}) holds.
\end{lemma}
\begin{proof} 
The set $Z$ is the image of the set $\cal S$ under the map $I:~(\varphi,\pi,q,p)\mapsto q$.
This map is continuous ${\cal E}_F\to \R^3$ and it is injection on ${\cal S}$.
Therefore ${\cal S}$ is a trapping set in ${\cal E}_F$, because $Z$ is a trapping set in $\R^3$. 
Hence (\re{WP9}) implies (\re{WP10}).
\end{proof}
 \setcounter{equation}{0}
 \section{Linearization at stationary state}\label{lin-sec}
In the rest of the paper we prove Theorem \re{WPC}.
If the particle is close to a stable minimum of $V$, we expect the nonlinear
evolution to be dominated by the linearized dynamics.
In this case the rate of the convergence  (\re{WP10}) corresponds to the decay rate of initial fields.
For notational simplicity we  assume isotropy in the following sense 
\begin{equation}\la{WPhesse}
 \partial_i\partial_j V(q_+)=\nu^2_0\delta_{ij},~~i,j=1,2,3~~ \nu_0>0\,.
\end{equation}
Without loss of generality we take $q_+=0$.
Let $S_0 = (s_0, 0, 0, 0)$ be the stationary state of (\re{WP2}) corresponding to $q_+=0$.
To linearize (\re{WP2}) at $S_0$, we set $\varphi(x,t)=s_0(x)+\psi(x,t)$. 
Then (\re{WP2}) becomes
 \beqn\la{WPnonlin}
\left. \ba{lll}
 \dot{\psi}(x, t)=\pi (x, t), ~~~~~~\dot{\pi}(x, t) &= &\Delta \psi(x, t)+\rho(x)-\rho(x-q(t)),\\
 \dot q(t) =  p(t),~~~~~~~~~~~~~~~~\dot{p}(t)  &= & -\nabla V(q(t))+\ds\int d^{\,3}x\,\,\psi(x, t)\,\nabla\rho(x-q(t)) \\
& &+\,\ds\int d^{\,3}x\,\,s_0(x)[\nabla\rho (x-q(t))-\nabla\rho(x)]\,.
\ea\right|
\eeqn
We denote $X(t)=Y(t)-S_0=(\psi(t),\pi(t),q(t),p(t))\in C(\R, {\cal E})$ and rewrite the nonlinear system (\re{WPnonlin}) in the form
\be\la{WPAB}
 \dot X(t)=AX(t)+B(X(t)).
\ee
Here the linear operator $A$ reads
$$
 A:(\psi,\pi,q,p)
 \mapsto (\pi,~\Delta\psi+\nabla\rho\cdot q,~p,~-(\nu^2_0+\nu^2_1)q+\int d^3x\,\psi(x)\nabla\rho(x)),
 $$
with
\be\la{WPO1}
 \nu^2_1\delta_{ij}=\frac{1}{3}\,\Vert\rho\Vert_{L^2}^2\delta_{ij}=-\,\int d^{\,3}x\,\pa_i s_0(x)\pa_j\rho(x).
 \ee
The factor $1/3$ is due to a spherical symmetry of $\rho(x)$ (see (\re{C})).
 The nonlinear part is given by
 \beqn\la{WPBB}
 B(X)\!\!\!& = &\!\!\!\Big(0,~\,\rho(x)\!-\!\rho(x-\!q)\!-\!\nabla\rho(x)\cdot q,~\, 0,
 ~-\nabla V(q)+\nu_0^2 q+\!\int d^{\,3}x\,\psi(x) [\nabla\rho(x-\!q)\!-\!\nabla\rho (x)]
 \nonumber\\ & & \,\,\,+\,\int d^{\,3}x \,\nabla s_0(x)
 [\rho(x)\!-\!\rho(x-\!q)\!-\!\nabla\rho (x)\cdot q]\,\Big).
\eeqn
Consider the Cauchy problem for the linear equation
\be\la{WPlin}
\dot Z(t)=AZ(t),\quad Z=(\Psi,\Pi,Q,P),\quad t\in\R,
\ee
with initial condition
\be\la{WPic}
Z|_{t=0}=Z_0=(\Psi_0,\Pi_0,Q_0,P_0).
\ee
System (\re{WPlin}) is a formal Hamiltonian system with the quadratic Hamiltonian
\be\la{WPH0}
{\cal H}_0(Z)=\fr 12 \Big( {P^2} +(\nu_0^2+\nu_1^2) Q^2+
\int d^3x\,(|\Pi(x)|^2+|\nabla\Psi(x)|^2-2\Psi(x)\nabla\rho(x)\cdot Q) \Big),
\ee
which is the formal Taylor expansion of ${\cal H}(Y_0+Z)$ up to second order at $Z=0$.
\begin{lemma}\la{WPexlin}
 Let condition (\re{C}) holds and  $Z_0\in{\cal E}$. Then\\
 (i) The Cauchy problem (\re{WPlin}), (\re{WPic}) has a unique solution $Z(t)\in C(\R,{\cal E})$.\\
 (ii) For every $t$, the map $U_0(t):Z_0\mapsto Z(t)$ is continuous both on  
${\cal E}$ and ${\cal E}_F$.\\
 (iii) The energy ${\cal H}_0$ is conserved, i.e.
 \be\la{WPec}
 {\cal H}_0(Z(t))={\cal H}_0(Z_0),\quad t\in\R.
 \ee
 iv) The estimate holds
 \be\la{WPlinb}
 \Vert Z(t)\Vert_{\cal E} \leq C,\quad t\in\R
 \ee
 with $C$ depending only on the norm $\Vert Z_0\Vert_{\cal E}$ of the initial state.
 \end{lemma}
 The key role in the proof is played  the positivity of  the Hamiltonian (\re{WPH0}):
 $$
 2{\cal H}_0(Z)={P^2} + \nu_0^2 Q^2+\int d^3x\,(|\Pi(x)|^2+
 |\nabla\Psi(x)+\rho(x) Q)|^2\geq 0.
 $$
 Thus (\re{WPlinb}) follows from (\re{WPec}) because of $\nu_0>0$.
The positivity of ${\cal H}_0$ is also obvious from (\re{WP1r}).

In \cite{KK16} we proved the following  long-time decay of the linearized dynamics
in the weighted Sobolev norms.
\begin{pro}\la{TDL}
Let conditions (\re{C})--(\re{W}) hold, and let $Z_0\in{\cal E}_\al$ with some $\al>1$.  Then 
\be\la{Z-dec}
\Vert U_0(t)Z_0\Vert_{-\al}\le C(\rho,\al)(1+|t|)^{-\al}\Vert Z_0\Vert_{\al},
\qquad t\in\R.
\ee
\end{pro}
Similar decay also holds for the dynamical group  $W(t)$ of 3D free wave equation. 
\begin{pro}\la{TDW} (cf. \cite[Proposition 2.1]{3w}
and \ci{KK2012}. 
Let $(\vp_0,\pi_0)\in{\cal F}_\al$ with some $\al>1$.  Then 
\be\la{WZ-dec}
\Vert W(t)(\vp_0,\pi_0)\Vert_{-\al}\le C(\al)(1+|t|)^{-\al}\Vert (\vp_0,\pi_0)\Vert_{\al},
\qquad t\in\R.
\ee
\end{pro}
We will use both these decays in the next section.

\section{A nonlinear Huygens  principle}\label{nH-sec}
\setcounter{equation}{0}
The following lemma is a version of strong Huygens principle for the nonlinear system (\re{WP2}).
 Let $M_*$ be a fixed number, $M_*>3R_\rho +1$.
\begin{lemma}\la{WPl9.1}
Let  conditions of Theorem \re{WPC} hold and let $\delta>0$ be an arbitrary fixed number. 
Then for sufficiently large $t_*>0$ there exists a solution
$$
 Y_*(t)=(\varphi_*(x,t),\pi_*(x,t),q_*(t),p_*(t))\-\in C([t_*,\infty),~{\cal E})
$$
to the system (\re{WP2}) such that \\
(i) $Y_*(t)$ coincides with $Y  (t)$ in some future cone,
\beqn\la{WP9.1}
\ba{rll}
\varphi_*(x,t)&=\varphi(x,t)& ~~~for~~|x|<t-t_*,\\
q_*(t)&=q(t)     & ~~~for~~t>t_*.
\ea
\eeqn
(ii) $Y_*(t_{*})$ admits a decomposition
$Y_*(t_*)=S_0+K_0+Z_0$, where $Z_0=(\Psi_0,\Pi_0,Q_0,P_0)$ satisfies
\be\la{WPCS'}
 \Psi_0(x)=\Pi_0(x)=0~~for~~|x|\geq M_*~,
\ee
\be\la{WP9.3}
\Vert Z_0\Vert_{\al}\leq \delta,
\ee
and $K_0$ satisfies
\be\la{WPA1}
 \Vert U_0(\tau)K_0\Vert_{-\al}\leq C(1+t_*+\tau)^{-\al},\quad \tau>0,
\ee
where $C=C(\al)$ does not depend on $\de$.
\end{lemma}
\begin{proof}
The convergence (\re{WPconv1}) with $q_+=0$ implies that for every $\epsilon>0$ 
there exist $t_{\epsilon}$ such that
\be\la{WP9.6}
 |q(t)|+|\dot q(t)|<\epsilon {\rm ~~~for~~} t>t_{\epsilon}.
\ee
We may assume that $t_{\epsilon}>1/{\epsilon}$. Denote
\be\la{WP9.6'}
t_{0,\epsilon}=t_{\epsilon}+R_\rho,
~~t_{1,\epsilon}=t_{0,\epsilon}+1,
~~t_{2,\epsilon}=t_{1,\epsilon}+\epsilon+R_\rho,
~~t_{3,\epsilon}=t_{2,\epsilon}+\epsilon+R_\rho.
\ee
 Then there exist a function $q_\epsilon(\cdot)\in C^1(\R)$ such that
 \beqn\la{WP9.7}
 q_\epsilon(t)=
 \left\{\ba{rl}
 q(t),&t>t_{1,\epsilon},\\
 0,&t<t_{0,\epsilon},
 \ea\right.{\rm ~~~~and~~}
 |q_\epsilon(t)|+|\dot q_\epsilon(t)|<\epsilon
 {\rm ~~~for~~all~~} t\in\R
\eeqn
 by suitable interpolation. Now  we define the modification
$\varphi_\ve(x,t)$ of the solution  $\varphi(x,t)=\varphi_{r}(x,t)+\varphi_K(x,t)$:
\be\la{WP5.11'}
 \varphi_\epsilon(x,t)=\varphi_{r,\ve}(x,t)+\varphi_{K}(x,t){\rm ~~for~~}x\in\R^3{\rm ~~and~~}t>0,
 \ee
 where
 \be\la{WP5.12}
 \varphi_{r,\epsilon}(x,t)=
 -\int~d^3y~\fr {1}{4\pi|x-y|}\rho(y-q_\epsilon(t-|x-y|))~.
 \ee
For $|x|<t-t_{2,\epsilon}$ and $|y|\le R_\rho+\epsilon$, we have
$$
t-|x-y|>t-(|x|+|y|)>t-(t-t_{2,\epsilon}+R_\rho+\epsilon)=t_{1,\epsilon}.
$$
Then  (\re{WP5.12}), (\re{WPretp}), and (\re{WP9.7}) imply
\be\la{WP5.14}
 \varphi_{r,\epsilon}(x,t)=
 \varphi_r(x,t)~~{\rm for~~}|x|<t-t_{2,\epsilon}.
\ee
Further, for $|x|>t-t_\epsilon$ and $|y|\le R_\rho$, we obtain
$$
t-|x-y|<t-(|x|-|y|)<t-(t-t_\epsilon-R_\rho)=t_{0,\epsilon}.
$$
Then $q_\epsilon(t-|x-y|)=0$ by (\re{WP9.7}), and hence
\be\la{WP9.12}
 \varphi_{r,\epsilon}(x,t)=s_0(x)~~{\rm for~~}|x|>t-t_\epsilon.
\ee
Moreover, $\varphi_{r,\epsilon}(\cdot,\cdot)\in C^1(\R^4)$, and (\re{WP9.7}) implies
\be\la{WP9.13}
 \sup_{x\in\R^3,\,t\in\R}(|\dot\varphi_{r,\epsilon}(x,t)|+
 |\nabla\varphi_{r,\epsilon}(x,t)-\nabla s_0(x)|+|\varphi_{r,\epsilon}(x,t)-s_0(x)|)={\cal O}(\epsilon).
\ee
 Now we  define $t_*:=t_{3,\epsilon}$, and
 \be\la{WP9.14}
 Y_*(t)=(\varphi_\epsilon(t),\dot\varphi_\epsilon(t),q(t), p(t)),~~
 K_0=(\varphi_K(t_*),\dot\varphi_K(t_*),0,0),~~
 Z_0=(\varphi_{r,\epsilon}(t_*)-s_0, \dot\varphi_{r,\epsilon}(t_*), q(t_*),p(t_*)).
\ee
It is easy to check that $t_{*}$ and $Y_*(t)$, $K_0$, $Z_0$
satisfy all requirements of Lemma \re{WP9.1}, provided $\epsilon>0$ be sufficiently small.

First,  $Y_*(x,t)$ is a solution to (\re{WP2}) for $t>t_*$.
Indeed, for $|x|<\epsilon+R_\rho$ one has $t-|x-y|>t_{3,\epsilon}-2\epsilon-2R_\rho=t_{1,\epsilon}$.
Since, (\re{WP9.6'}) implies that $q_\epsilon(t-|x-y|)=q(t-|x-y|)$ and
$\varphi_\epsilon(x,t)= \varphi(x,t)$ then.
Hence, $Y_*(t)=Y(t)$ in the region $|x|<\epsilon+R_\rho$.
On the other hand, (\re{WP9.6}) and (\re{WP9.7}) imply
$$
 \rho(x-q_\epsilon(t))= \rho(x-q(t))= 0{\rm ~~for~~}|x|>\epsilon+R_\rho
 {\rm ~~and~~}t>t_\epsilon.
$$
Hence, $\varphi_{r,\epsilon}(x,t)$ satisfies the equation 
\be\la{seq}
\ddot \varphi(x,t)=\Delta \varphi(x,t){\rm ~~for~~}|x|>\epsilon+R_\rho
 {\rm ~~and~~}t>t_\epsilon.
\ee
Therefore,  $Y_{*}(t)$ satisfies  (\re{WP2}) in the region $|x|>\epsilon+R_\rho$. Now
(\re{WP9.1}) follows from (\re{WP9.7}) and (\re{WP5.14}), 
(\re{WPCS'}) for $M_*=3R_\rho+2\epsilon+1$ follows from (\re{WP9.12}), and
(\re{WP9.3}) follows from (\re{WPCS'}) and (\re{WP9.13}).

It remains to prove (\re{WPA1}). We deduce the estimate from the decay (\re{WZ-dec})
for the linearized dynamics $U(t)$  and decay (\re{Z-dec}) for $W(t)$.
Denote  $U(\tau)K_0=(\Psi(x,\tau),\Pi(x,\tau), Q(\tau), P(\tau))$.
From \cite [formulas (4.18), (4.19), (4.25)]{KK16} it follows that
$$
\left(\ba{c} Q(\tau) \\ P(\tau)\ea\right)
={\cal L}*\left(\ba{c}
0 \\ f_k\ea\right)(\tau)
$$
where 
$$
f_k(\tau)=\langle W(\tau)[\phi_k(t_*),\dot\phi_k(t_*)],\na\rho\rangle
=\langle W(\tau+t_*)[\phi_0,\pi_0],\na\rho\rangle,
$$
and
$$
(\Psi(\tau),\Pi(\tau))=W(\tau+t_*)[\phi_0,\pi_0]+\int_0^\tau W(\tau-s)[0, Q(s)\cdot\na\rho]ds
$$
Moreover, according to \cite [formula (4.20)]{KK16} for ${\cal L}(t)$ the decay holds
$$
{\cal L}(t)={\cal O}(|t|)^{-N},\quad t\to\infty,\quad\forall N>0.
$$
Then the decay (\re{WPA1}) follows.
\end{proof}

 \setcounter{equation}{0}
 \section{The rate of convergence }\label{rc-sec}
Here we prove  Theorem \re{WPC} i). Due to (\re{WP9.1}) it suffices to prove that for any $\ve>0$
\be\la{WPA2}
 \Vert Y_*(t)-S_0\Vert_{-\al}={\cal O}(t^{-\al+\ve}),\quad t\to\infty.
\ee
Denote
$X(\tau)=Y_*(t_*+\tau)-S_0$. Then $X(0)=K_0+Z_0$ and the integrated version of  (\re{WPAB}) reads
\begin{equation}\la{WPkolb'}
 X(\tau)=U_0(\tau)K_0+U_0(\tau)Z_0+\int^\tau_0 ds\,U_0(\tau-s) B(X(s)),\quad\tau >0.
\end{equation}
Further,  (\re{Z-dec}), (\re{WPBB}), (\re{WPCS'}) and (\re{WPA1})  imply 
\be\la{WPiin'}
 {\|X(\tau)\|}_{-\al} \le C\bigg((t_*+\tau+1)^{-\al}+(1+\tau)^{-\al}{\|Z_0\|}_{\al} + \int^\tau_0 ds \,(1+\tau-s)^{-\al}
 {\|X(s)\|}^2_{-\al}\bigg),\quad\tau >0.
\ee
We fix an arbitrary  $\ve\in (0,1/2)$ and introduce the majorant
\be\la{maj}
m(t)=\sup_{0\leq s\leq t}(1+s)^{\al-\ve}\Vert X(s)\Vert_{-\al}.
\ee
Let $\mu$ be any fixed positive number, and  let $T_\mu$ be the exit time:
\be\la{T}
T_\mu=\sup \{t>0:m(t)\le \mu\}.
\ee
Multiplying both sides of (\re{WPiin'}) by $(1+\tau)^{\al-\ve}$,
and taking the supremum in $\tau\in[0,T_\mu]$, we get
\be\la{WPiin1'}
 m(\tau)\le C\bigg(\fr{(1+\tau)^{\al-\ve}} {(1+t_*+\tau)^{\al}}+
  \de + \int^\tau_0 ds\fr{(1+\tau)^{\al-\ve}}{(1+\tau-s)^{\al}}\fr{m^2(s)}
{(1+s)^{2\al-2\ve}}\bigg),\quad \tau\le T_\mu.
 \ee
Note  that for every $\ve>0$
\be\la{WPt*}
 \sup_{\tau>0}\fr{(1+\tau)^{\al-\ve}} {(1+t_*+\tau)^{\al}}
\to 0,\quad t_*\to\infty.
\ee
Hence taking into account that $m(t)$ is a monotone
increasing function, we get for sufficiently large $t_*$ that
\be\la{mest}
m(\tau)\le C(\de+Cm^2(\tau)),\quad\tau\le T_\mu.
\ee
This inequality implies that $m(\tau)$ is bounded for $\tau\le T_\mu$, and moreover,
\be\la{m2est}
m(\tau)\le C_1\de,\quad\tau\le T
\ee
if $\de$ is sufficiently small.
The constant $C_1$ in  (\re{m2est}) does not depend on $T$.
Due to  Lemma \re{WPl9.1}  we can choose $t_*$   so large that $\de<\mu/(2C_1)$.
Then (\re{m2est}) implies that  $T=\infty$ and (\re{m2est}) holds for all $\tau >0$ if $t_*$ is sufficiently large.
{\hfill $\Box$}

\section{Scattering asymptotics}\label{as-sec}
\setcounter{equation}{0}
Here we prove Theorem \re{WPC} ii). 
We prove  asymptotics  (\re{S})--(\re{rm}) for $t\to+\infty$ only since system (\re{WP2}) is time reversible. 
Denote $\Phi(x,t)=(\Phi_1(x,t),\Phi_2(x,t))=(\varphi(x,t),\pi(x,t))-(s_{q_+},0)$. 
Then asymptotics (\re{S})-- (\re{rm}) are equivalent to
\[ 
\Phi(t)=W(t)\Phi_++r(t),\quad \Vert r(t)\Vert_{\Ho^1\oplus L^2}=
{\cal O}(t^{-\al+1+\ve}),\quad t \to+\infty,
\]
This is equivalent to  
\be\la{Sme} 
W(-t)\Phi(t)=\Phi_+ +r_1(t),~~~~~ 
\Vert r_1(t)\Vert_{\Ho^1\oplus L^2}={\cal O}(t^{-\al+1+\ve}), \quad t \to+\infty
\ee 
due to the unitarity of  $W(t)$  on  $\Ho^1\oplus L^2$. 
The first two equations of (\re{WP2}) imply  
$$ 
\dot \Phi_1(x,t)=\Phi_2(x,t),\quad \dot \Phi_2(x,t)=\Delta \Phi_1(x,t)+\rho(x-q_+)-\rho(x-q(t)).
$$ 
Then 
\be\la{eqacc} 
\Phi(t)=W(t)\Phi(0)-\int_0^tW(t-s)[(0,\rho(x-q_+)-\rho(x-q(s)))]ds. 
\ee 
Therefore,
\be\la{duhs} 
W(-t)\Phi(t)= \Phi(0)-\int_0^t W(-s)R(s)ds,\qquad R(s)=(0,\rho(x-q_+)-\rho(x-q(s)),  
\ee 
where the integral   converges in  $\Ho^1\oplus L^2$ with the rate ${\cal O}(t^{-\al+1+\ve})$. Indeed, 
\[
\Vert  W(-s)R(s)\Vert_{\Ho^1\oplus L^2} ={\cal O}(s^{-\al+\ve}),\quad 0<\ve<\al-1
\]
by the unitarity of $W(-s)$ and the decay rate 
$\Vert R(s)\Vert_{\Ho^1\oplus L^2}={\cal O}(s^{-\al+\ve})$ 
which follows   from  the conditions (\re{C}) on $\rho$ and the asymptotics (\re{WPEE1}).
Setting
\[
\Phi_+=\Phi(0)-\int_0^\infty W(-s)R(s)ds,\quad r_1(t)=\int_t^\infty W(-s)R(s)ds,
\]
we obtain (\re{Sme}).

\end{document}